\documentclass{llncs}
\usepackage{amssymb,nicefrac}
\usepackage{amsfonts}
\usepackage{graphicx}
\usepackage[fleqn]{amsmath}
\usepackage[varg]{txfonts} 
\usepackage{stmaryrd}

\usepackage{framed}
\usepackage{color}
\usepackage{ulem}
\usepackage{tikzfig}
\tikzstyle{env}=[copoint,regular polygon rotate=0,minimum width=0.2cm, fill=black]

\tikzstyle{every picture}=[baseline=-0.25em]
\tikzstyle{dotpic}=[scale=0.5]
\tikzstyle{diredges}=[every to/.style={diredge}]
\tikzstyle{dot graph}=[shorten <=-0.1mm,shorten >=-0.1mm,scale=0.6]
\tikzstyle{plot point}=[circle,fill=black,minimum width=2mm,inner sep=0]

\tikzstyle{braceedge}=[decorate,decoration={brace,amplitude=2mm,raise=-1mm}]
\tikzstyle{small braceedge}=[decorate,decoration={brace,amplitude=1mm,raise=-1mm}]
\tikzstyle{left hook arrow}=[left hook-latex]
\tikzstyle{right hook arrow}=[right hook-latex]

\tikzstyle{dtriangle}=[fill=yellow,draw=black,shape=isosceles triangle,shape border rotate=-90,isosceles triangle stretches=true,inner sep=1pt,minimum width=0.4cm,minimum height=3mm]
\tikzstyle{vtriang}=[fill=yellow,draw=black,shape=isosceles triangle,shape border rotate=180,isosceles triangle stretches=true,inner sep=1pt,minimum width=0.4cm,minimum height=3mm]
\tikzstyle{triangle}=[fill=yellow,draw=black,shape=isosceles triangle,shape border rotate=90,isosceles triangle stretches=true,inner sep=1pt,minimum width=0.4cm,minimum height=3mm]
\tikzstyle{H box}=[rectangle,fill=yellow,draw=black,xscale=1.0,yscale=1.0, inner sep=1.pt]
\tikzstyle{gbox}=[rectangle,fill=green,draw=black,xscale=1.0,yscale=1.0, inner sep=1.pt]
\tikzstyle{rbox}=[rectangle,fill=red,draw=black,xscale=1.0,yscale=1.0, inner sep=1.pt]

\tikzstyle{bn}=[circle,fill=black,draw=black,scale=.4]
\tikzstyle{wn}=[circle,fill=white,draw=black,scale=.6]
\tikzstyle{dn}=[circle,fill=none,draw=gray]

\tikzstyle{black dot}=[inner sep=0.7mm,minimum width=0pt,minimum height=0pt,fill=black,draw=black,shape=circle]

\tikzstyle{dot}=[black dot]
\tikzstyle{smalldot}=[inner sep=0.4mm,minimum width=0pt,minimum height=0pt,fill=black,draw=black,shape=circle]%NEW
\tikzstyle{white dot}=[dot,fill=white]
\tikzstyle{antipode}=[white dot,inner sep=0.3mm,font=\footnotesize]
\tikzstyle{smallwhitedot}=[smalldot,fill=white]%NEW
\tikzstyle{alt white dot}=[white dot,label={[xshift=3.07mm,yshift=-0.05mm,font=\footnotesize]left:$*$}]
\tikzstyle{gray dot}=[dot,fill=gray!40!white]
\tikzstyle{smallgraydot}=[smalldot,fill=gray!40!white]%NEW
\tikzstyle{box vertex}=[draw=black,rectangle]
\tikzstyle{small box}=[box vertex,fill=white]%% added rwd]
\tikzstyle{whitebg}=[fill=white,inner sep=2pt]
\tikzstyle{graph state vertex}=[sg vertex,fill=black]

\tikzstyle{wide copoint}=[fill=white,draw=black,shape=isosceles triangle,shape border rotate=90,isosceles triangle stretches=true,inner sep=1pt,minimum width=1.5cm,minimum height=5mm]
\tikzstyle{wide point}=[fill=white,draw=black,shape=isosceles triangle,shape border rotate=-90,isosceles triangle stretches=true,inner sep=1pt,minimum width=1.5cm,minimum height=4mm]
\tikzstyle{very wide copoint}=[fill=white,draw=black,shape=isosceles triangle,shape border rotate=-90,isosceles triangle stretches=true,inner sep=1pt,minimum width=2.5cm,minimum height=4mm]
\tikzstyle{very wide empty copoint}=[draw=black,shape=isosceles triangle,shape border rotate=-90,isosceles triangle stretches=true,inner sep=1pt,minimum width=2.5cm,minimum height=4mm]
\tikzstyle{symm}=[ultra thick,shorten <=-1mm,shorten >=-1mm]

\tikzstyle{small box}=[rectangle,inline text,fill=white,draw,minimum height=5mm,yshift=-0.5mm,minimum width=5mm,font=\small]
\tikzstyle{square box}=[rectangle,fill=white,draw=black,minimum height=5mm,minimum width=5mm,font=\small]
\tikzstyle{square gray box}=[rectangle,fill=gray!30,draw=black,minimum height=6mm,minimum width=6mm]
\tikzstyle{copoint}=[regular polygon,regular polygon sides=3,draw=black,scale=0.75,inner sep=-0.5pt,minimum width=7mm,fill=white]
\tikzstyle{point}=[regular polygon,regular polygon sides=3,draw=black,scale=0.75,inner sep=-0.5pt,minimum width=7mm,fill=white,regular polygon rotate=180]
\tikzstyle{gray point}=[point,fill=gray!40!white]
\tikzstyle{gray copoint}=[copoint,fill=gray!40!white]

\newcommand{\edgearrow}{{\arrow[black]{>}}}
\newcommand{\edgetick}{{\arrow[black,scale=0.7,very thick]{|}}}
\tikzstyle{diredge}=[->]
\tikzstyle{rdiredge}=[<-]
\tikzstyle{medium diredge}=[->]

\tikzstyle{short diredge}=[->]
\tikzstyle{halfedge}=[-)]
\tikzstyle{other halfedge}=[(-]
\tikzstyle{freeedge}=[(-)]
\tikzstyle{white edge}=[line width=5pt,white]
\tikzstyle{tick}=[postaction=decorate,decoration={markings, mark=at position 0.5 with \edgetick}]
\tikzstyle{small map edge}=[|-latex, gray!60!blue, shorten <=0.9mm, shorten >=0.5mm]
\tikzstyle{thick dashed edge}=[very thick,dashed,gray!40]
\tikzstyle{map edge}=[|-latex,very thick, gray!40, shorten <=1mm, shorten >=0.5mm]
\tikzstyle{tickedge}=[postaction=decorate,
  decoration={markings, mark=at position 0.5 with \edgetick}]
\tikzstyle{dirtickedge}=[postaction=decorate,
  decoration={markings, mark=at position 0.5 with \edgetick},
  decoration={markings, mark=at position 0.85 with \edgearrow}]
\tikzstyle{dirdoubletickedge}=[postaction=decorate,
  decoration={markings, mark=at position 0.4 with \edgetick},
  decoration={markings, mark=at position 0.6 with \edgetick},
  decoration={markings, mark=at position 0.85 with \edgearrow}]

\makeatletter
\newcommand{\boxshape}[3]{%
\pgfdeclareshape{#1}{
\inheritsavedanchors[from=rectangle] % this is nearly a rectangle
\inheritanchorborder[from=rectangle]
\inheritanchor[from=rectangle]{center}
\inheritanchor[from=rectangle]{north}
\inheritanchor[from=rectangle]{south}
\inheritanchor[from=rectangle]{west}
\inheritanchor[from=rectangle]{east}
\backgroundpath{% this is new
\southwest \pgf@xa=\pgf@x \pgf@ya=\pgf@y
\northeast \pgf@xb=\pgf@x \pgf@yb=\pgf@y

\@tempdima=#2
\@tempdimb=#3

\pgfpathmoveto{\pgfpoint{\pgf@xa - 5pt + \@tempdima}{\pgf@ya}}
\pgfpathlineto{\pgfpoint{\pgf@xa - 5pt - \@tempdima}{\pgf@yb}}
\pgfpathlineto{\pgfpoint{\pgf@xb + 5pt + \@tempdimb}{\pgf@yb}}
\pgfpathlineto{\pgfpoint{\pgf@xb + 5pt - \@tempdimb}{\pgf@ya}}
\pgfpathlineto{\pgfpoint{\pgf@xa - 5pt + \@tempdima}{\pgf@ya}}
\pgfpathclose
}
}}

\boxshape{NEbox}{0pt}{8pt}
\boxshape{SEbox}{0pt}{-8pt}
\boxshape{NWbox}{8pt}{0pt}
\boxshape{SWbox}{-8pt}{0pt}
\makeatother

\tikzstyle{map}=[draw,shape=NEbox,inner sep=7pt]
\tikzstyle{mapdag}=[draw,shape=SEbox,inner sep=7pt]
\tikzstyle{maptrans}=[draw,shape=SWbox,inner sep=7pt]
\tikzstyle{mapconj}=[draw,shape=NWbox,inner sep=7pt]
\tikzstyle{inline text}=[text height=1.5ex, text depth=0.25ex,yshift=0.5mm]
\tikzstyle{medium box}=[rectangle,inline text,fill=white,draw,minimum height=3.5mm,yshift=-0.5mm,minimum width=10mm,font=\small]
\tikzstyle{semilarge box}=[rectangle,inline text,fill=white,draw,minimum height=3.5mm,yshift=-0.5mm,minimum width=12.5mm,font=\small]
\tikzstyle{right label}=[label,anchor=west,xshift=-1.5mm]

\tikzstyle{probs}=[shape=semicircle,fill=gray!40!white,draw=black,shape border rotate=180,minimum width=1.2cm]

\tikzstyle{arrs}=[-latex,font=\small,auto]
\tikzstyle{arrow plain}=[arrs]
\tikzstyle{arrow dashed}=[dashed,arrs]
\tikzstyle{arrow bold}=[very thick,arrs]
\tikzstyle{arrow hide}=[draw=white!0,-]
\tikzstyle{arrow reverse}=[latex-]
\tikzstyle{cdnode}=[]

\tikzstyle{gn}=[dot,fill=green,minimum width=0.2cm,inner sep=0pt]
\tikzstyle{rn}=[dot,fill=red,inner sep=0pt,minimum width=0.2cm]
\tikzstyle{rc}=[dot,thick,fill=white,draw = red,minimum width=0.3cm,inner sep=0pt]
\tikzstyle{gc}=[dot,thick,fill=white,draw= green,inner sep=0pt,minimum width=0.3cm]
\tikzstyle{bc}=[dot,thick,fill=white,draw= blue,minimum width=0.3cm]

\tikzstyle{label}=[circle,fill=white,minimum width=0.3cm]

\tikzstyle{clocklabel}=[dot,fill=yellow,draw=black,font=\tiny,inner sep=0.75pt]

\tikzstyle{rsn}=[circle split,draw,fill=red,font=\tiny,inner sep=0.75pt]
\tikzstyle{gsn}=[circle split,draw,fill=green,font=\tiny,inner sep=0.75pt]
\tikzstyle{bsn}=[circle split,draw,fill=blue,font=\tiny,inner sep=0.75pt]

\tikzstyle{rsc}=[circle split,thick,draw= red,draw,fill=white,font=\tiny,inner sep=0.75pt]
\tikzstyle{gsc}=[circle split,thick,draw= green,draw,fill=white,font=\tiny,inner sep=0.75pt]
\tikzstyle{bsc}=[circle split,thick,draw= blue,draw,fill=white,font=\tiny,inner sep=0.75pt]

\tikzstyle{cnot}=[fill=white,shape=circle,inner sep=-1.4pt]
\tikzstyle{wire label}=[font=\tiny, auto]

\newcommand{\bra}[1]{\ensuremath{\left\langle #1 \right|}}
\newcommand{\ket}[1]{\ensuremath{\left|  #1 \right\rangle}}

\tikzstyle{cdiag}=[matrix of math nodes, row sep=3em, column sep=3em, text height=1.5ex, text depth=0.25ex,inner sep=0.5em]
\tikzstyle{arrow above}=[transform canvas={yshift=0.5ex}]
\tikzstyle{arrow below}=[transform canvas={yshift=-0.5ex}]

\def\bR{\begin{color}{red}}  
\def\bB{\begin{color}{blue}}  
\def\bM{\begin{color}{magenta}} 
\def\bGr{\begin{color}{darkgray}}
\def\bC{\begin{color}{cyan}}
\def\bW{\begin{color}{white}}
\def\bBl{\begin{color}{black}}
\def\bG{\begin{color}{green}}
\def\bY{\begin{color}{yellow}}
\def\jR{\begin{color}{magenta}}
\def\jB{\begin{color}{cyan}}
\def\e{\end{color}}

\usepackage{mathtools}
\usepackage{placeins}  
\title{ZX-Rules for 2-qubit Clifford+T Quantum Circuits}
\author{Bob Coecke \and Quanlong Wang}
\institute{University of Oxford \\
\email{\{Bob.Coecke, Quanlong.Wang\}@cs.ox.ac.uk}}  

\begin{document} 
\maketitle

\begin{abstract}
 ZX-calculus is a high-level graphical formalism for qubit computation. In this paper we give the ZX-rules that enable one to derive all equations between 2-qubit Clifford+T quantum circuits.  Our rule set is only a small extension of the rules of stabiliser ZX-calculus, and substantially less than those needed for the recently achieved universal completeness.  One of our rules is new, and we expect it to also have other utilities.   

These ZX-rules are  much simpler than the complete of set Clifford+T circuit equations due to Selinger and Bian, which indicates that ZX-calculus  provides a more convenient arena for quantum circuit rewriting than restricting oneself to circuit equations.  The reason for this is that ZX-calculus is not constrained by a fixed unitary gate set for performing intermediate computations.  
%In this paper, we prove the completeness of the ZX-calculus ( with just 8 rules) for 2-qubit Clifford+T circuits by verifying the complete set of 17 circuit relations in diagrammatic rewriting. As a consequence, we generalised Backen's result of completeness of  the ZX-calculus for single-qubit Clifford+T group. We are also able to give an analytic solution for converting from ZXZ to XZX Euler decompositions of single-qubit unitary gates as suggested by Schr\"oder de Witt and Zamdzhiev. Taking account of the universal completeness of the ZX-calculus for pure qubit quantum mechanics, the result in this paper is an important step towards efficient simplification of arbitrary n-qubit  Clifford+T circuits.
\end{abstract} 

\section{Introduction}

The ZX-calculus \cite{CD1,CD2} 
is a universal graphical language for qubit  theory, which comes equipped with simple rewriting rules that enable one to transform diagrams representing one quantum process into  another quantum process. More broadly, it is part of categorical quantum mechanics which aims for a high-level formulation of quantum theory \cite{AC1,CKbook}. It has found applications both in quantum foundations \cite{CDKZ,CDKZ2,MiriamSpek} and quantum computation \cite{DP2,Clare,DomError,de2017zx}, 
and is subject to automation thanks to the Quantomatic software \cite{quanto-cade}.   Recently ZX-calculus has been completed by Ng and Wang \cite{ng2017universal}, that is, provided with sufficient additional rules so that any equation between matrices in Hilbert space can be derived in ZX-calculus.  This followed  earlier completions  by Backens   for  stabiliser theory \cite{Backens} and one-qubit Clifford+T circuits \cite{Backens2},  and   by Jeandel, Perdrix and Vilmart for general Clifford+T theory \cite{jeandel2017complete}. In Section \ref{sec:ZX-rules} we present Backens' two theorems.
 
This paper concerns a sufficient set of  ZX-rules for establishing all equations between 2-qubit Clifford+T quantum circuits, which  again can be seen as a completeness result.  
%\bR For that purpose it suffices to use the ZX-rules ignoring non-zero scalars. \e 
We were motivated in two manners to seek  this result: 
\begin{itemize}
\item Firstly, we wish to understand the utility of the ZX-rules.  In the case of the full completion \cite{ng2017universal,jeandel2018diagrammatic} these were added using a purely theoretical methodology which consisted of translating Hilbert space structure into diagrams, passing via another graphical calculus \cite{Amar,hadzihasanovic2017algebra}. 
 However, a natural question concerns  the actual practical use of each of these rules, as well as of other rules derived from them.   As an example, one of the key ZX-rules: 
 \[
\beginpgfgraphicnamed{diagrams/b2s}
\InputIfFileExists{diagrams/b2s.tikz}{}{\input{./figures/diagrams/b2s.tikz}}
\endpgfgraphicnamed
 \]
is equivalent to the following well known circuit equation \cite{CD2}: 
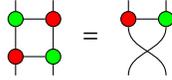
 \[
\beginpgfgraphicnamed{diagrams/strongcomplementary1CNOT}
\InputIfFileExists{diagrams/strongcomplementary1CNOT.tikz}{}{\input{./figures/diagrams/strongcomplementary1CNOT.tikz}}
\endpgfgraphicnamed
 \]
 involving CNOT gates (green $\simeq$ control).  In this paper we are concerned with  all such equations for 2-qubit Clifford+T quantum circuits.
\item Secondly, in quantum computing algorithms are converted into elementary gates making up circuits, and these circuits then have to be implemented on a computer. Currently the most considered universal set of elementary gates is the Clifford+T gate set. The high cost of implementing those gates makes  any simplification of a circuit (cf.~having less CNOT-gates and/or having less T-gates) highly desirable.  We expect our result to be an important stepping stone towards efficient simplification of arbitrary n-qubit  Clifford+T circuits, and that the quantomatic software will be a crucial part of this. The fact that a small set of rules suffices for us here raises the hope that general circuit simplification could already be done with a small set of ZX-rules.
\end{itemize}

Selinger and Bian derived a complete set of circuit equations for Clifford+T 2-qubit circuits \cite{ptbian}.
However, these circuit equations are very large and rigid, and their method for producing these beyond two-qubits doesn't scale to more qubits.  On the other hand, in the case of ZX-calculus we already have an overarching completeness results that carries over to circuits of arbitrary qubits.  So the main question then concerns the rules needed specifically for efficient circuit rewriting.
%\bR Here somewhere(?): Unlike in the single-qubit case \cite{x}, there is no algorithm \bR yet(?) \e that is both optimal and efficient for multi-qubit unitary synthesis. \e 

The advantage of ZX-rules is that they are not constrained by unitarity.  Also, in the ZX computation at intermediate stages phase gates may not even be within Clifford+T, although their actual values play no roles, that is, they can be treated as variables.  Note that going beyond the constraints of the formalism which one aims to prove something about is a standard practice in mathematics, e.g.~complex analysis.

%\bR 
%comparing to the normal ZX rules \cite{CoeckeDuncan, bpw}, we just added the (P) rule, which is a property of the general Euler decomposition in ZXZ and XZX forms. To give an analytic solution for converting from ZXZ to XZX Euler decompositions of single-qubit unitary gates is called by Schr\"oder de Witt and Zamdzhiev in \cite{Zamdzhiev}. Here we first give an explicit formula for the relation of  ZXZ tand XZX Euler decompositions of generalised Z and X phases, then obtain as a corollary the formula for normal Z and X phases. Yet we do not need the explicit formula to verify those circuit relations, what is really useful is just a property ((P) rule), which makes it much simpler for simplifying circuits. However, we pay at the price of go beyond the range of Clifford+T, i.e., angles other than multiple of $\pi/4$ are involved. 
%\e

\section{Background 1: ZX-calculus language} 

A pedestrian introduction is \cite{coecke2012tutorial}.  There are two ways to present ZX-calculus, either as \emph{diagrams} or as a \emph{category}.  Following \cite{CKbook}, the `language' of the ZX-calculus consists of certain special \emph{processes} or \emph{boxes}:
 \[
\beginpgfgraphicnamed{diagrams//box}
\InputIfFileExists{diagrams//box.tikz}{}{\input{./figures/diagrams//box.tikz}}
\endpgfgraphicnamed
 \]
 which can be wired together to form \emph{diagrams}:
 \[
\beginpgfgraphicnamed{diagrams//compound-process-capscups}
\InputIfFileExists{diagrams//compound-process-capscups.tikz}{}{\input{./figures/diagrams//compound-process-capscups.tikz}}
\endpgfgraphicnamed
 \]
All the diagrams should be read from top to bottom. Note that the wiring of inputs to inputs and outputs to inputs, as well as feed-back loops is admitted.  Equivalently, following \cite{CD2}, it consists of certain morphisms in a compact closed category, which has the natural numbers: $0, 1, 2,  \cdots$   as objects, with the addition of numbers as the tensor:
\[
m \otimes n = m+n 
\] 
%...the ZX-calculus makes a compact closed category $\mathfrak{C}$. The objects of $\mathfrak{C}$ are natural numbers: $0, 1, 2,  \cdots$; the tensor of objects is just addition of numbers: $m \otimes n = m+n$. The morphisms of $\mathfrak{C}$ are diagrams of the ZX-calculus. A general diagram  $D:k\to l$   with $k$ inputs and $l$ outputs is generated by:
In diagrams $n$ corresponds to $n$ wires side-by-side.

The special processes/boxes/morphisms that we are concerned with in this paper are \emph{spiders} of two \emph{colours}:
 \[
\beginpgfgraphicnamed{diagrams//spider_green_alpha}
\InputIfFileExists{diagrams//spider_green_alpha.tikz}{}{\input{./figures/diagrams//spider_green_alpha.tikz}}
\endpgfgraphicnamed\qquad%
\beginpgfgraphicnamed{diagrams//spider_red_alpha}
\InputIfFileExists{diagrams//spider_red_alpha.tikz}{}{\input{./figures/diagrams//spider_red_alpha.tikz}}
\endpgfgraphicnamed
 \]
where $\alpha\in[0, 2\pi)$. Equivalently, one can only consider spiders of one colour as well as a colour changer  (cf.~rule (H2) below):
 \[
\beginpgfgraphicnamed{diagrams//HadaDecomSingleslt}
\begin{tikzpicture}
	\begin{pgfonlayer}{nodelayer}
		\node [style=H box] (0) at (-0.75, 0) {$H$};
		\node [style=none] (1) at (-0.75, -0.25) {};
		\node [style=none] (2) at (-0.75, 0.25) {};
	\end{pgfonlayer}
	\begin{pgfonlayer}{edgelayer}
		\draw (2.center) to (0);
		\draw (1.center) to (0);
	\end{pgfonlayer}
\end{tikzpicture}}
\endpgfgraphicnamed  
 \]
ZX-calculus can also be seen as a calculus of graphs, provided that one introduces special input and output nodes. 

Sometimes it is useful to also think of wires appearing in the diagram as boxes, which can take the following forms:
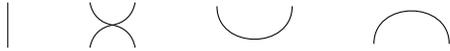
\[
\beginpgfgraphicnamed{diagrams//Id}
\begin{tikzpicture}
	\begin{pgfonlayer}{nodelayer}
		\node [style=none] (1) at (0.5, 0.3) {};
		\node [style=none] (2) at (0.5, -0.3) {};
		\node [style=none] (3) at (0.5, -0.5) {};
		\node [style=none] (4) at (0.5, 0.5) {};
	\end{pgfonlayer}
	\begin{pgfonlayer}{edgelayer}
		\draw (1.center) to (2.center);
	\end{pgfonlayer}
\end{tikzpicture}}
\endpgfgraphicnamed\qquad\quad%
\beginpgfgraphicnamed{diagrams//swap}
\InputIfFileExists{diagrams//swap.tikz}{}{\input{./figures/diagrams//swap.tikz}}
\endpgfgraphicnamed\qquad\quad%
\beginpgfgraphicnamed{diagrams//cup}
\begin{tikzpicture}
	\begin{pgfonlayer}{nodelayer}
		\node [style=none] (0) at (-0.5, 0.25) {};
		\node [style=none] (1) at (0.5, 0.25) {};
	\end{pgfonlayer}
	\begin{pgfonlayer}{edgelayer}
		\draw [bend right=90, looseness=1.50] (0.center) to (1.center);
	\end{pgfonlayer}
\end{tikzpicture}}
\endpgfgraphicnamed\qquad\quad %
\beginpgfgraphicnamed{diagrams//cap}
\begin{tikzpicture}
	\begin{pgfonlayer}{nodelayer}
		\node [style=none] (0) at (-0.5, -0.25) {};
		\node [style=none] (1) at (0.5, -0.25) {};
	\end{pgfonlayer}
	\begin{pgfonlayer}{edgelayer}
		\draw [bend left=90, looseness=1.50] (0.center) to (1.center);
	\end{pgfonlayer}
\end{tikzpicture}}
\endpgfgraphicnamed
\]
In particular, then the full specification of what `wiring boxes together' actually means can be reduced to what it means to put boxes side-by-side and connect the output of a box to the input of another box:
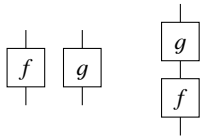
\[
\beginpgfgraphicnamed{diagrams//box_par}
\InputIfFileExists{diagrams//box_par.tikz}{}{\input{./figures/diagrams//box_par.tikz}}
\endpgfgraphicnamed\qquad\raisebox{1mm}{%
\beginpgfgraphicnamed{diagrams//box_seq}
\InputIfFileExists{diagrams//box_seq.tikz}{}{\input{./figures/diagrams//box_seq.tikz}}
\endpgfgraphicnamed}
\]
The following key property uses this fact: 

\begin{theorem}\cite{ContPhys,CD2}
The ZX language is \emph{universal} for qubit quantum computing, when giving the following interpretation:
\[
\left\llbracket %
\beginpgfgraphicnamed{diagrams//generator_spider_alpha}
\InputIfFileExists{diagrams//generator_spider_alpha.tikz}{}{\input{./figures/diagrams//generator_spider_alpha.tikz}}
\endpgfgraphicnamed \right\rrbracket=\ket{0}^{\otimes m}\bra{0}^{\otimes n}+e^{i\alpha}\ket{1}^{\otimes m}\bra{1}^{\otimes n}  \qquad 
\left\llbracket%
\beginpgfgraphicnamed{diagrams//HadaDecomSingleslt}
\begin{tikzpicture}
	\begin{pgfonlayer}{nodelayer}
		\node [style=H box] (0) at (-0.75, 0) {$H$};
		\node [style=none] (1) at (-0.75, -0.25) {};
		\node [style=none] (2) at (-0.75, 0.25) {};
	\end{pgfonlayer}
	\begin{pgfonlayer}{edgelayer}
		\draw (2.center) to (0);
		\draw (1.center) to (0);
	\end{pgfonlayer}
\end{tikzpicture}}
\endpgfgraphicnamed\right\rrbracket=\frac{1}{\sqrt{2}}\begin{pmatrix}
        1 & 1 \\
        1 & -1
 \end{pmatrix}\\
 \]
 \[ 
\left\llbracket%
\beginpgfgraphicnamed{diagrams//Id}
\begin{tikzpicture}
	\begin{pgfonlayer}{nodelayer}
		\node [style=none] (1) at (0.5, 0.3) {};
		\node [style=none] (2) at (0.5, -0.3) {};
		\node [style=none] (3) at (0.5, -0.5) {};
		\node [style=none] (4) at (0.5, 0.5) {};
	\end{pgfonlayer}
	\begin{pgfonlayer}{edgelayer}
		\draw (1.center) to (2.center);
	\end{pgfonlayer}
\end{tikzpicture}}
\endpgfgraphicnamed\right\rrbracket=
\begin{pmatrix}
        1 & 0 \\
        0 & 1
 \end{pmatrix}
 \qquad 
   \left\llbracket%
\beginpgfgraphicnamed{diagrams//swap}
\InputIfFileExists{diagrams//swap.tikz}{}{\input{./figures/diagrams//swap.tikz}}
\endpgfgraphicnamed\right\rrbracket=
   \begin{pmatrix}
        1 & 0 & 0 & 0 \\
        0 & 0 & 1 & 0 \\
        0 & 1 & 0 & 0 \\
        0 & 0 & 0 & 1 
 \end{pmatrix}  
 \qquad 
  \left\llbracket%
\beginpgfgraphicnamed{diagrams//cap}
\begin{tikzpicture}
	\begin{pgfonlayer}{nodelayer}
		\node [style=none] (0) at (-0.5, -0.25) {};
		\node [style=none] (1) at (0.5, -0.25) {};
	\end{pgfonlayer}
	\begin{pgfonlayer}{edgelayer}
		\draw [bend left=90, looseness=1.50] (0.center) to (1.center);
	\end{pgfonlayer}
\end{tikzpicture}}
\endpgfgraphicnamed\right\rrbracket=
  \begin{pmatrix}
        1  \\
        0  \\
        0  \\
        1  \\
 \end{pmatrix}
 \qquad 
   \left\llbracket%
\beginpgfgraphicnamed{diagrams//cup}
\begin{tikzpicture}
	\begin{pgfonlayer}{nodelayer}
		\node [style=none] (0) at (-0.5, 0.25) {};
		\node [style=none] (1) at (0.5, 0.25) {};
	\end{pgfonlayer}
	\begin{pgfonlayer}{edgelayer}
		\draw [bend right=90, looseness=1.50] (0.center) to (1.center);
	\end{pgfonlayer}
\end{tikzpicture}}
\endpgfgraphicnamed\right\rrbracket=
   \begin{pmatrix}
        1 & 0 & 0 & 1 
         \end{pmatrix}
\]         
\[
  \left\llbracket %
\beginpgfgraphicnamed{diagrams//box_par}
\InputIfFileExists{diagrams//box_par.tikz}{}{\input{./figures/diagrams//box_par.tikz}}
\endpgfgraphicnamed  \right\rrbracket =  \left\llbracket %
\beginpgfgraphicnamed{diagrams//box_f}
\InputIfFileExists{diagrams//box_f.tikz}{}{\input{./figures/diagrams//box_f.tikz}}
\endpgfgraphicnamed  \right\rrbracket \otimes  \left\llbracket  %
\beginpgfgraphicnamed{diagrams//box_g}
\InputIfFileExists{diagrams//box_g.tikz}{}{\input{./figures/diagrams//box_g.tikz}}
\endpgfgraphicnamed  \right\rrbracket  \qquad 
 \left\llbracket \raisebox{1mm}{%
\beginpgfgraphicnamed{diagrams//box_seq}
\InputIfFileExists{diagrams//box_seq.tikz}{}{\input{./figures/diagrams//box_seq.tikz}}
\endpgfgraphicnamed}  \right\rrbracket =  \left\llbracket %
\beginpgfgraphicnamed{diagrams//box_f}
\InputIfFileExists{diagrams//box_f.tikz}{}{\input{./figures/diagrams//box_f.tikz}}
\endpgfgraphicnamed  \right\rrbracket \circ  \left\llbracket  %
\beginpgfgraphicnamed{diagrams//box_g}
\InputIfFileExists{diagrams//box_g.tikz}{}{\input{./figures/diagrams//box_g.tikz}}
\endpgfgraphicnamed  \right\rrbracket
  \]
 That is, every linear map of type $\mathbb{C}^{2^n}\to\mathbb{C}^{2^m}$  can be written down as a ZX-diagram, and consequently, every qubit  process can be written down as a ZX-diagram.
\end{theorem}

\section{Background 2: ZX-calculus rules}\label{sec:ZX-rules}

 Above we specified the ingredients of the ZX-calculus as linear maps.  Now, in quantum theory linear maps only matter up to a non-zero scalar multiple, i.e.~a diagram with no inputs nor outputs.  We will do so too here, since this makes that the rules of the ZX-calculus appear much simpler (see e.g.~\cite{backens2017towards} for a presentation of the ZX-calculus rules with explicit scalars that make equations hold on-the-nose).   

 Due to the diagrammatic underpinning, in addition to the rules given below, there is one meta-rule that ZX-calculus obeys, namely:
\begin{center}
\fbox{\it Only connectedness matters!}
\end{center}
One could do without it by adding a few more rules, but it is entirely within the spirit of diagrammatic reasoning that it should all boil down to connectedness. We now give an overview of ZX-rule sets that have been considered. %\bR Note that all the ZX rules in this paper hold up to a non-zero scalar. \e

\emph{Stabiliser ZX-calculus} is the restriction of ZX-calculus to $\alpha\in\{{n\pi\over 2}\mid n\in\mathbb{N}\}$.  As shown in \cite{Backens}, the following rules make ZX-calculus complete for this fragment of quantum theory:
\[ 
\begin{tabular}{ccccc}
\beginpgfgraphicnamed{diagrams/spider-bis}
\InputIfFileExists{diagrams/spider-bis.tikz}{}{\input{./figures/diagrams/spider-bis.tikz}}
\endpgfgraphicnamed&\quad(S1)&$\qquad$&%
\beginpgfgraphicnamed{diagrams/greenredidentity}
\InputIfFileExists{diagrams/greenredidentity.tikz}{}{\input{./figures/diagrams/greenredidentity.tikz}}
\endpgfgraphicnamed&\quad(S2)\\
  &&\ \ && \\
\beginpgfgraphicnamed{diagrams/b1}
\InputIfFileExists{diagrams/b1.tikz}{}{\input{./figures/diagrams/b1.tikz}}
\endpgfgraphicnamed&\quad(B1)&&%
\beginpgfgraphicnamed{diagrams/b2s}
\InputIfFileExists{diagrams/b2s.tikz}{}{\input{./figures/diagrams/b2s.tikz}}
\endpgfgraphicnamed&\quad(B2)\\
  &&\ \ && \\
\beginpgfgraphicnamed{diagrams/hadamdecom}
\InputIfFileExists{diagrams/hadamdecom.tikz}{}{\input{./figures/diagrams/hadamdecom.tikz}}
\endpgfgraphicnamed&\quad(H1)&&%
\beginpgfgraphicnamed{diagrams/clchge}
\InputIfFileExists{diagrams/clchge.tikz}{}{\input{./figures/diagrams/clchge.tikz}}
\endpgfgraphicnamed&\quad(H2)
\end{tabular}
\]

That is, any equation between stabiliser ZX-diagrams that can be proven using matrices can also be proved by using these rules. 

 The `only connectivity matters rule' means that  we also have \cite{backens2017towards}:
\[
\beginpgfgraphicnamed{diagrams/induced_compact_structure}
\InputIfFileExists{diagrams/induced_compact_structure.tikz}{}{\input{./figures/diagrams/induced_compact_structure.tikz}}
\endpgfgraphicnamed\hfill(S2')
\]
%Two
Some  other derivable rules that we will use are: 
\[
\beginpgfgraphicnamed{diagrams/Hopfx}
\InputIfFileExists{diagrams/Hopfx.tikz}{}{\input{./figures/diagrams/Hopfx.tikz}}
\endpgfgraphicnamed   \quad(Hf)  \qquad %
\beginpgfgraphicnamed{diagrams//hexagon2}
\InputIfFileExists{diagrams//hexagon2.tikz}{}{\input{./figures/diagrams//hexagon2.tikz}}
\endpgfgraphicnamed\quad(Hex) \qquad   %
\beginpgfgraphicnamed{diagrams/Cy}
\InputIfFileExists{diagrams/Cy.tikz}{}{\input{./figures/diagrams/Cy.tikz}}
\endpgfgraphicnamed   \quad (Cy) 
\] 
where  the dots in (Cy) denote zero or more wires. The 1st and last rule are derived in  \cite{CD2} and the middle one in \cite{DP1}.
%\TODOb{We should give these.}
% \bR The $\pi$ copy rule  and the Hopf law  are also derivable from these rules   \cite{backens2017towards}, we will use the two properties in this paper without presentation of their diagrams. \e 
We also use the following  variation form of (B2), to which we also refer as (B2): % but with the same name for convenience:
\[ 
\beginpgfgraphicnamed{diagrams/b2var}
\InputIfFileExists{diagrams/b2var.tikz}{}{\input{./figures/diagrams/b2var.tikz}}
\endpgfgraphicnamed \hfill(B2)
\] 
The rules (S1) and (H) apply to spiders with an arbitrary number of input and output wires, including none,  so (S1) and (H) appear to be an infinite set of rules. Firstly, these rules do have algebraic counterparts as Frobenius algebras, which constitute a finite set.  Secondly, using the concept of \emph{bang-boxes} \cite{kissinger2016tensors}, even in their present form these rules can be notationally reduced  to a single rule, and the quantomatic-software accounts for rules in this form.  Allowing for bang-boxes, one can  also merge rules (B1) and (B2) into a single rule:
\[
\beginpgfgraphicnamed{diagrams/strongcomplementaryn}
\InputIfFileExists{diagrams/strongcomplementaryn.tikz}{}{\input{./figures/diagrams/strongcomplementaryn.tikz}}
\endpgfgraphicnamed 
\]
hence reducing the number of equations to be memorised to six.

%An alternative  more  version of the stabiliser ZX-rules, taken from \cite{CKBook}, is:
%\[
%\bR PIC\e
%\]
%Besides being slightly more economical, here the rules have a clear connotation: xxx represents how spiders of some colour interact with themselves, xxx represents how they interact with each other, and the remaining rule xxx is a specific rule about the geometry of the Block sphere, in particular that besides the Z- and the X-observables there is also a Y-observable that can be defined both in terms of the Z- and the X-observable.

\emph{Single-qubit Clifford+T ZX-calculus} is the restriction of ZX-calculus to spiders with exactly one input and one output, and $\alpha\in\{{n\pi\over 4}\mid n\in\mathbb{N}\}$.  As shown in \cite{Backens2}, the  rules (S1), (S2), (H1) and (H2) together with the rule:
\[
\beginpgfgraphicnamed{diagrams/k2}
\InputIfFileExists{diagrams/k2.tikz}{}{\input{./figures/diagrams/k2.tikz}}
\endpgfgraphicnamed\hfill(N)
\]
make ZX-calculus complete for this fragment of quantum theory.  We will  also use the following special form of the (N) rule, to which we again refer as (N):  
\[ 
\beginpgfgraphicnamed{diagrams/nvar}
\InputIfFileExists{diagrams/nvar.tikz}{}{\input{./figures/diagrams/nvar.tikz}}
\endpgfgraphicnamed \hfill(N)
\]

As single qubit circuits can be seen as a restriction of 2-qubit circuits, simply by letting the 2nd qubit unaltered, our result can also be seen as a completeness result for single-qubit Clifford+T ZX-calculus.   However, it is weaker than Backens' as we employ more rules.  

%\begin{center} % \tikzfig{diagrams//emptysquare-small}
%\begin{tabular}{|r@{~}r@{~}c@{~}c|r@{~}r@{~}c@{~}c|}
%\hline
%$R_Z^{(n,m)}$&$:$&$n\to m$ & \tikzfig{diagrams//generator_spider} & $A$&$:$&$ 1\to 1$& \tikzfig{diagrams//alphagate}\\
%\hline
%$H$&$:$&$1\to 1$ &\tikzfig{diagrams//HadaDecomSingleslt}
%%& $s$&$:$&$0\to 0$ &\tikzfig{scalars//halfscalar}  &
% &  $\sigma$&$:$&$ 2\to 2$& \tikzfig{diagrams//swap}\\\hline
%   $\mathbb I$&$:$&$1\to 1$&\tikzfig{diagrams//Id} & $e $&$:$&$0 \to 0$& \tikzfig{diagrams//emptysquare}\\\hline
%   $C_a$&$:$&$ 0\to 2$& \tikzfig{diagrams//cap} &$ C_u$&$:$&$ 2\to 0$&\tikzfig{diagrams//cup} \\\hline
%\end{tabular}
%\end{center}
%where $m,n\in \mathbb N$, $\alpha \in [0,  2\pi)$, and $e$ represents an empty diagram. 

\section{Result: ZX rules vs.~circuit equations.}
 
 Recall that in this paper  the ZX-rules  hold up to a non-zero scalar.
%, so are the equalities that will be derived in this section. 
%This is our main theorem: 
 
\begin{theorem}\label{eq:mainthm}
The  rules (S1), (S2), %(S3), 
(B1), (B2), (H1), (H2), (N) and (P) depicted below make ZX-calculus complete for 2-qubit Clifford+T circuits:
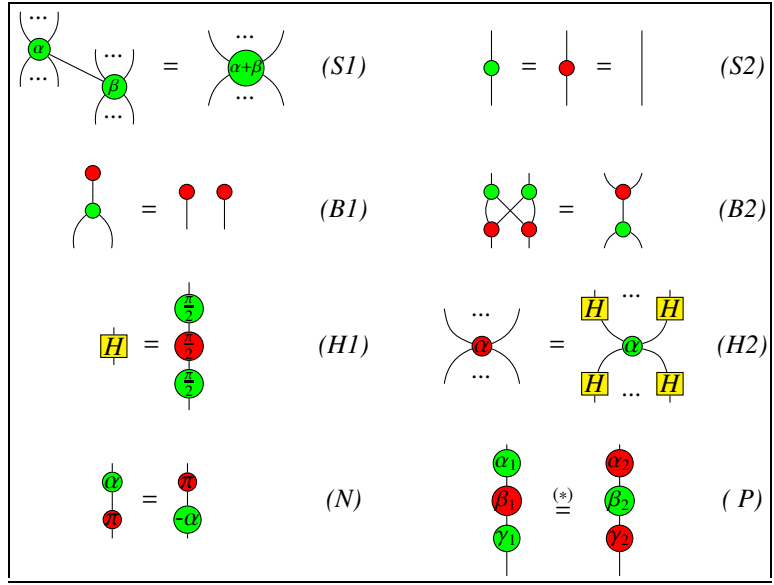
\[ 
  \begin{tabular}{|ccccc|}
  \hline
\beginpgfgraphicnamed{diagrams/spider-bis}
\InputIfFileExists{diagrams/spider-bis.tikz}{}{\input{./figures/diagrams/spider-bis.tikz}}
\endpgfgraphicnamed&\quad(S1)&$\qquad$&%
\beginpgfgraphicnamed{diagrams/greenredidentity}
\InputIfFileExists{diagrams/greenredidentity.tikz}{}{\input{./figures/diagrams/greenredidentity.tikz}}
\endpgfgraphicnamed&\quad(S2)\\
  &&\ \ && \\
\beginpgfgraphicnamed{diagrams/b1}
\InputIfFileExists{diagrams/b1.tikz}{}{\input{./figures/diagrams/b1.tikz}}
\endpgfgraphicnamed&\quad(B1)&&%
\beginpgfgraphicnamed{diagrams/b2s}
\InputIfFileExists{diagrams/b2s.tikz}{}{\input{./figures/diagrams/b2s.tikz}}
\endpgfgraphicnamed&\quad(B2)\\
  &&\ \ && \\
\beginpgfgraphicnamed{diagrams/hadamdecom}
\InputIfFileExists{diagrams/hadamdecom.tikz}{}{\input{./figures/diagrams/hadamdecom.tikz}}
\endpgfgraphicnamed&\quad(H1)&&%
\beginpgfgraphicnamed{diagrams/clchge}
\InputIfFileExists{diagrams/clchge.tikz}{}{\input{./figures/diagrams/clchge.tikz}}
\endpgfgraphicnamed&\quad(H2)\\
  &&\ \ && \\
%  \tikzfig{diagrams/k2}&\quad(N)&&\tikzfig{diagrams/clswapaa}&\quad(P)\vspace{0.5mm}
 %
\beginpgfgraphicnamed{diagrams/k2}
\InputIfFileExists{diagrams/k2.tikz}{}{\input{./figures/diagrams/k2.tikz}}
\endpgfgraphicnamed&\quad(N)&& %
\beginpgfgraphicnamed{diagrams/zxztoxzx2}
\InputIfFileExists{diagrams/zxztoxzx2.tikz}{}{\input{./figures/diagrams/zxztoxzx2.tikz}}
\endpgfgraphicnamed&\quad( P)\vspace{0.5mm}
\\
  \hline
\end{tabular}
  \]
%\[
% \begin{tabular}{|ccccc|}
%  \hline
%  &&&& \\
%  \tikzfig{diagrams/spider-bis}&\ \ (S1)&$\qquad$&\tikzfig{diagrams/induced_compact_structure}&\ \ (S2)\\
%  &&&& \\
% % \tikzfig{diagrams/b1}&\ \ (B1)&&\tikzfig{diagrams/b2s}&\ \ (B2)\\
%  \tikzfig{diagrams/redidentity}\quad(S3)& \tikzfig{diagrams/b1}&(B1)&\tikzfig{diagrams/b2s}&(B2)\\
%  &&&& \\
%  \tikzfig{diagrams/hadamdecom}&\ \ (H1)&&\tikzfig{diagrams/clchge}&\ \ (H2)\\
%  &&&& \\
%   \tikzfig{diagrams/k2}&\ \ (N)&&\tikzfig{diagrams/clswapaa}&\ \ (P)\\
%  &&&& \\
%  \hline
% \end{tabular}
% \]
 %where  $s\in\{0,1\}$, and 
 where  $\alpha_2=\gamma_2$ if $\alpha_1=\gamma_1$, and  $\alpha_2=\pi+\gamma_2$ if $\alpha_1=-\gamma_1$; the equality (*) should be read as follows: for every diagram in LHS there exists $\alpha_2, \beta_2$ and $\gamma_2$ such that LHS=RHS (and vice versa if conjugating by the Hadamard gate).  In what follows we will see that we actually don't need to know the precise values of $\alpha_2, \beta_2$ and $\gamma_2$. 
 \end{theorem} 
 
So as compared to the rules that we saw in the previous section there is only one additional rule here, the (P) rule.  This rule is a new rule that was not present as such in any previous presentation of the ZX-calculus.  Of course, as the rules presented in \cite{ng2017universal} yield universal completeness,  one should be able to derive it from these:
%or better, one should be able to derive the following (purely equational) rule from these:

\begin{lemma}\label{zxztoxzxcr}
%For \bR normal $Z$ and $X$ phases\e,  
For $\alpha_1, \beta_1, \gamma_1 \in (0, ~2\pi)$ we have:
\begin{equation}\label{zxztoxzxcreq}
\beginpgfgraphicnamed{diagrams//zxztoxzx}
\InputIfFileExists{diagrams//zxztoxzx.tikz}{}{\input{./figures/diagrams//zxztoxzx.tikz}}
\endpgfgraphicnamed\qquad\mbox{with}\quad
\left\{
\begin{array}{l}
\alpha_2=\arg z+\arg z'\\
\beta_2=2\arg (|\frac{z}{z'}|+i)\\
\gamma_2=\arg z-\arg z'
\end{array}
\right.
\end{equation}
where:
\[
\begin{array}{l}
z=\cos\frac{\beta_1}{2}\cos\frac{\alpha_1+\gamma_1}{2}+i\sin\frac{\beta_1}{2}\cos\frac{\alpha_1-\gamma_1}{2}
\qquad
z'=\cos\frac{\beta_1}{2}\sin\frac{\alpha_1+\gamma_1}{2}-i\sin\frac{\beta_1}{2}\sin\frac{\alpha_1-\gamma_1}{2}
\end{array}
\]
So if $\alpha_1=\gamma_1$, then $\alpha_2=\gamma_2$, and if $\alpha_1=-\gamma_1$, %then $\alpha_2=\pi+\gamma_2(Mod ~2\pi)$.
 then $\alpha_2=\pi+\gamma_2$. 
\end{lemma}

%The proof of this Lemma is in the appendix, 
 This Lemma is restated as Corollary \ref{zxztoxzxcr} and proved in the appendix, which has a more general analytic solution for this `colour-swapping' property for arbitrary generalised phases.   
%This Lemma, proved in the arXiv version of this paper \cite{coecke2018circuit}, gives an analytic solution for the `colour-swapping' property for arbitrary phases. 
The idea for the need for a rule of this kind was first suggested by Schr\"oder de Witt and Zamdzhiev \cite{VladComp}.  As already indicated in the introduction, it is also clear that this rule takes one out of the Clifford+T realm  in the sense that the values of the angles in the RHS of (\ref{zxztoxzxcreq}) usually go beyond Clifford+T even if the LHS is inside of the realm.

The  proof of Theorem \ref{eq:mainthm} draws from Selinger and Bian's \cite{ptbian} set of circuit equations that  is complete for 2-qubit circuits.  Here we rely on universality of ZX-language to write down these circuits, and in particular  besides CNOT-gates these also involve symmetric CZ-gates:
\[
\beginpgfgraphicnamed{diagrams//CZ}
\InputIfFileExists{diagrams//CZ.tikz}{}{\input{./figures/diagrams//CZ.tikz}}
\endpgfgraphicnamed  
\]
In the statement of the following theorem we  adopt the more usual left-to-right reading of circuits, although we still express it as ZX diagrams.
 
\begin{theorem}\label{ptbianthm}\cite{ptbian} 
The following equations are complete for 2-qubit Clifford+T circuits: 
\begin{equation}\label{cmrns-2}
\beginpgfgraphicnamed{diagrams//completerelationlist-2}
\InputIfFileExists{diagrams//completerelationlist-2.tikz}{}{\input{./figures/diagrams//completerelationlist-2.tikz}}
\endpgfgraphicnamed
\end{equation}\vspace{-1.5mm}
\begin{equation}\label{cmrns-1}
\beginpgfgraphicnamed{diagrams//completerelationlist-1}
\InputIfFileExists{diagrams//completerelationlist-1.tikz}{}{\input{./figures/diagrams//completerelationlist-1.tikz}}
\endpgfgraphicnamed
\end{equation}\vspace{-0.5mm}
\begin{equation}\label{cmrns0}
\beginpgfgraphicnamed{diagrams//completerelationlist0}
\InputIfFileExists{diagrams//completerelationlist0.tikz}{}{\input{./figures/diagrams//completerelationlist0.tikz}}
\endpgfgraphicnamed
\end{equation}\vspace{0.0mm}
\begin{equation}\label{cmrns1}
\beginpgfgraphicnamed{diagrams//completerelationlist1}
\InputIfFileExists{diagrams//completerelationlist1.tikz}{}{\input{./figures/diagrams//completerelationlist1.tikz}}
\endpgfgraphicnamed
\end{equation}\vspace{1.5mm}
\begin{equation}\label{cmrns2}
\beginpgfgraphicnamed{diagrams//completerelationlist2}
\InputIfFileExists{diagrams//completerelationlist2.tikz}{}{\input{./figures/diagrams//completerelationlist2.tikz}}
\endpgfgraphicnamed
\end{equation}\vspace{1.5mm}
\begin{equation}\label{cmrns3}
\beginpgfgraphicnamed{diagrams//completerelationlist3}
\InputIfFileExists{diagrams//completerelationlist3.tikz}{}{\input{./figures/diagrams//completerelationlist3.tikz}}
\endpgfgraphicnamed
\end{equation}\vspace{1.5mm}
\begin{equation}\label{cmrns4}
\beginpgfgraphicnamed{diagrams//completerelationlist4}
\InputIfFileExists{diagrams//completerelationlist4.tikz}{}{\input{./figures/diagrams//completerelationlist4.tikz}}
\endpgfgraphicnamed
\end{equation}\vspace{1.5mm}
\begin{equation}\label{cmrns5}
\beginpgfgraphicnamed{diagrams//completerelationlist5}
\InputIfFileExists{diagrams//completerelationlist5.tikz}{}{\input{./figures/diagrams//completerelationlist5.tikz}}
\endpgfgraphicnamed
\end{equation}\vspace{1.5mm}
\begin{equation}\label{cmrns6}
\beginpgfgraphicnamed{diagrams//completerelationlist6}
\InputIfFileExists{diagrams//completerelationlist6.tikz}{}{\input{./figures/diagrams//completerelationlist6.tikz}}
\endpgfgraphicnamed
\end{equation}\vspace{1.5mm}
\begin{equation}\label{cmrns7}
\beginpgfgraphicnamed{diagrams//completerelationlist7}
\InputIfFileExists{diagrams//completerelationlist7.tikz}{}{\input{./figures/diagrams//completerelationlist7.tikz}}
\endpgfgraphicnamed
\end{equation}\vspace{1.5mm}
\begin{equation}\label{cmrns8}
\beginpgfgraphicnamed{diagrams//completerelationlist8}
\InputIfFileExists{diagrams//completerelationlist8.tikz}{}{\input{./figures/diagrams//completerelationlist8.tikz}}
\endpgfgraphicnamed
\end{equation}\vspace{1.5mm}
\begin{equation}\label{cmrns9}
\beginpgfgraphicnamed{diagrams//completerelationlist9}
\InputIfFileExists{diagrams//completerelationlist9.tikz}{}{\input{./figures/diagrams//completerelationlist9.tikz}}
\endpgfgraphicnamed
\end{equation}\vspace{1.5mm}
\begin{equation}\label{cmrns10}
\beginpgfgraphicnamed{diagrams//completerelationlist10}
\InputIfFileExists{diagrams//completerelationlist10.tikz}{}{\input{./figures/diagrams//completerelationlist10.tikz}}
\endpgfgraphicnamed
\end{equation}\vspace{1.5mm}
\begin{equation}\label{cmrns11}
\beginpgfgraphicnamed{diagrams//completerelationlist11}
\InputIfFileExists{diagrams//completerelationlist11.tikz}{}{\input{./figures/diagrams//completerelationlist11.tikz}}
\endpgfgraphicnamed 
\end{equation}\vspace{1.5mm}
\begin{equation}\label{cmrns12}
\left(%
\beginpgfgraphicnamed{diagrams//pbct1}
\InputIfFileExists{diagrams//pbct1.tikz}{}{\input{./figures/diagrams//pbct1.tikz}}
\endpgfgraphicnamed\right)^2=%
\beginpgfgraphicnamed{diagrams//completerelationlist12}
\begin{tikzpicture}
	\begin{pgfonlayer}{nodelayer}
		\node [style=none] (0) at (-0.25, 0.25) {};
		\node [style=none] (1) at (0.5, 0.25) {};
		\node [style=none] (2) at (-0.25, -0.25) {};
		\node [style=none] (3) at (0.5, -0.25) {};
	\end{pgfonlayer}
	\begin{pgfonlayer}{edgelayer}
		\draw (0.center) to (1.center);
		\draw (2.center) to (3.center);
	\end{pgfonlayer}
\end{tikzpicture}}
\endpgfgraphicnamed
\end{equation}\vspace{1.5mm}
\begin{equation}\label{cmrns13}
\left(%
\beginpgfgraphicnamed{diagrams//pbctb}
\InputIfFileExists{diagrams//pbctb.tikz}{}{\input{./figures/diagrams//pbctb.tikz}}
\endpgfgraphicnamed\right)^2=%
\beginpgfgraphicnamed{diagrams//completerelationlist12}
\begin{tikzpicture}
	\begin{pgfonlayer}{nodelayer}
		\node [style=none] (0) at (-0.25, 0.25) {};
		\node [style=none] (1) at (0.5, 0.25) {};
		\node [style=none] (2) at (-0.25, -0.25) {};
		\node [style=none] (3) at (0.5, -0.25) {};
	\end{pgfonlayer}
	\begin{pgfonlayer}{edgelayer}
		\draw (0.center) to (1.center);
		\draw (2.center) to (3.center);
	\end{pgfonlayer}
\end{tikzpicture}}
\endpgfgraphicnamed
\end{equation}\vspace{-1.5mm}
\begin{eqnarray*}
&&%
\beginpgfgraphicnamed{diagrams//completerelationlist141}
\InputIfFileExists{diagrams//completerelationlist141.tikz}{}{\input{./figures/diagrams//completerelationlist141.tikz}}
\endpgfgraphicnamed\vspace{1.5mm}\\
&&%
\beginpgfgraphicnamed{diagrams//completerelationlist142}
\InputIfFileExists{diagrams//completerelationlist142.tikz}{}{\input{./figures/diagrams//completerelationlist142.tikz}}
\endpgfgraphicnamed\vspace{1.5mm}\\ 
&&\ =%
\beginpgfgraphicnamed{diagrams//completerelationlist12}
\begin{tikzpicture}
	\begin{pgfonlayer}{nodelayer}
		\node [style=none] (0) at (-0.25, 0.25) {};
		\node [style=none] (1) at (0.5, 0.25) {};
		\node [style=none] (2) at (-0.25, -0.25) {};
		\node [style=none] (3) at (0.5, -0.25) {};
	\end{pgfonlayer}
	\begin{pgfonlayer}{edgelayer}
		\draw (0.center) to (1.center);
		\draw (2.center) to (3.center);
	\end{pgfonlayer}
\end{tikzpicture}}
\endpgfgraphicnamed
\end{eqnarray*}
\vspace{-13mm}
\begin{equation}\label{cmrns14}
\end{equation}
\end{theorem}
 
 Not only does this Theorem serve as a stepping stone, it  is also the 
% them 
 main point of comparison of our result. The ZX-rules are clearly much simpler than the circuit equations, which, to say the least, are virtually impossible to memorise, let alone apply.
 
 \section{Proof.}

We need to show that the  equations  in Theorem \ref{ptbianthm} can be derived from those in Theorem \ref{eq:mainthm}.  Doing so is a straightforward calculation for the first 14 ones.  However, this is not the case for the remaining  circuit relations (\ref{cmrns12}), (\ref{cmrns13}) and (\ref {cmrns14})  each of which we prove as a lemma. 

%Now we prove the three complicated  separately. %The proof of  is very similar to that of   (\ref{cmrns12}), so we place the poof in the appendix.

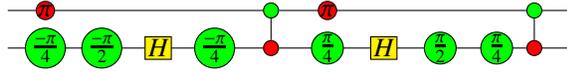
\begin{lemma}
Let $A=$
\[
\beginpgfgraphicnamed{diagrams//pbct1}
\InputIfFileExists{diagrams//pbct1.tikz}{}{\input{./figures/diagrams//pbct1.tikz}}
\endpgfgraphicnamed
\]
 then  $A^2=I.$
\end{lemma}

\begin{proof}
First we have $A=$
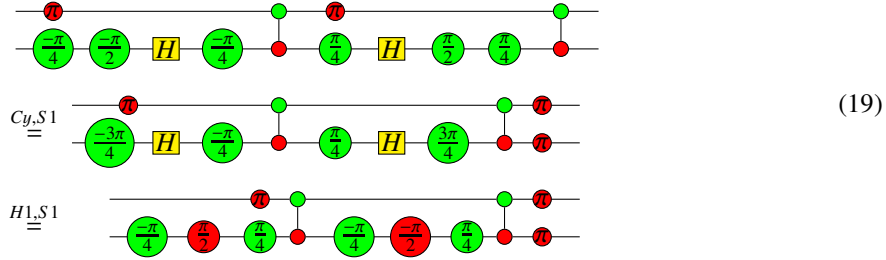
\begin{equation}
%\tikzfig{diagrams//pbct2}
%\tikzfig{diagrams//pbct22}
%
\beginpgfgraphicnamed{diagrams//pbct22new}
\InputIfFileExists{diagrams//pbct22new.tikz}{}{\input{./figures/diagrams//pbct22new.tikz}}
\endpgfgraphicnamed
\end{equation}
By the rule (P), we can assume that
\begin{equation}\label{pbct32}
\beginpgfgraphicnamed{diagrams//pbct32new}
\InputIfFileExists{diagrams//pbct32new.tikz}{}{\input{./figures/diagrams//pbct32new.tikz}}
\endpgfgraphicnamed
\end{equation}
Since $e^{i\frac{-\pi}{4}} e^{i\frac{\pi}{4}}=1 $, we could let $\gamma=\alpha+\pi$.
Also note that 
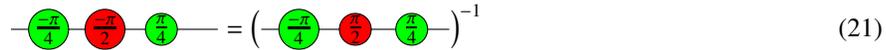
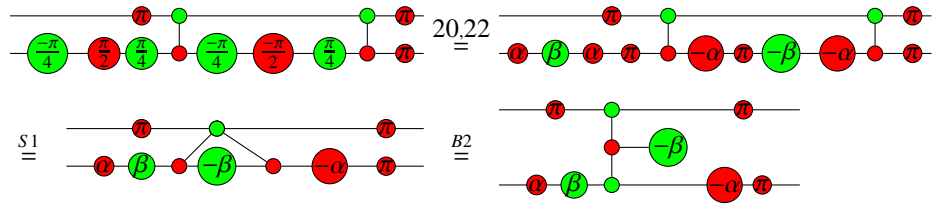
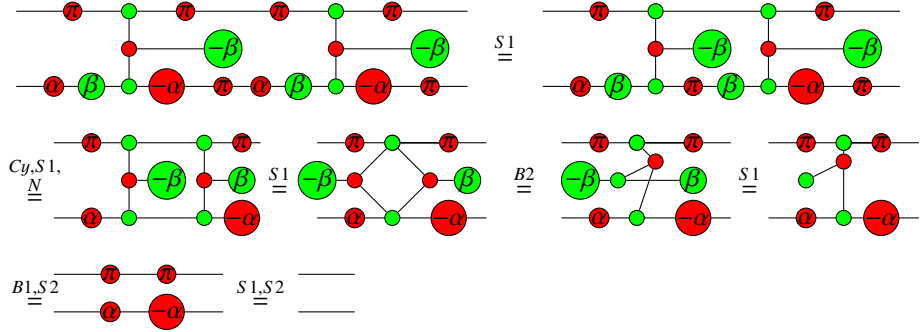
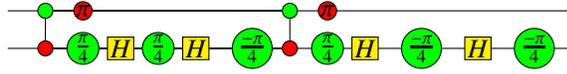
\begin{equation}
\beginpgfgraphicnamed{diagrams//pbct421new}
\InputIfFileExists{diagrams//pbct421new.tikz}{}{\input{./figures/diagrams//pbct421new.tikz}}
\endpgfgraphicnamed=\left(%
\beginpgfgraphicnamed{diagrams//pbct422new}
\InputIfFileExists{diagrams//pbct422new.tikz}{}{\input{./figures/diagrams//pbct422new.tikz}}
\endpgfgraphicnamed\right)^{-1}
\end{equation}
Thus:
\begin{equation}\label{pbct52}
\beginpgfgraphicnamed{diagrams//pbct52new}
\InputIfFileExists{diagrams//pbct52new.tikz}{}{\input{./figures/diagrams//pbct52new.tikz}}
\endpgfgraphicnamed
\end{equation}
Therefore, $A=$
$$
\beginpgfgraphicnamed{diagrams//pbct62}
\InputIfFileExists{diagrams//pbct62.tikz}{}{\input{./figures/diagrams//pbct62.tikz}}
\endpgfgraphicnamed
$$
Finally,
$A^2=$
$$
\beginpgfgraphicnamed{diagrams//pbct722}
\InputIfFileExists{diagrams//pbct722.tikz}{}{\input{./figures/diagrams//pbct722.tikz}}
\endpgfgraphicnamed
$$
\end{proof}

\begin{lemma}
Let $B=$
\[
\beginpgfgraphicnamed{diagrams//pbctb}
\InputIfFileExists{diagrams//pbctb.tikz}{}{\input{./figures/diagrams//pbctb.tikz}}
\endpgfgraphicnamed
\]
 then   $B^2=I$.
\end{lemma}

\begin{proof}
Firstly we have:
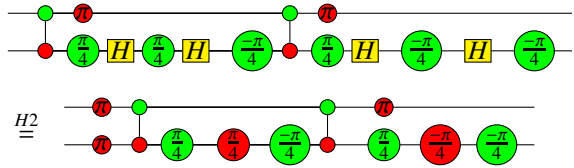
\[
%\tikzfig{diagrams//pbctb2}
%
\beginpgfgraphicnamed{diagrams//pbctb22new}
\InputIfFileExists{diagrams//pbctb22new.tikz}{}{\input{./figures/diagrams//pbctb22new.tikz}}
\endpgfgraphicnamed
\]
By the rule (P), we can assume that:
\begin{equation}\label{pbctb32}
\beginpgfgraphicnamed{diagrams//pbctb32new}
\InputIfFileExists{diagrams//pbctb32new.tikz}{}{\input{./figures/diagrams//pbctb32new.tikz}}
\endpgfgraphicnamed
\end{equation}
Since $e^{i\frac{-\pi}{4}} e^{i\frac{\pi}{4}}=1 $, we could let $\gamma=\alpha+\pi$.
% \bB and by self-inverseness\e:
Also note that: 
\[
\beginpgfgraphicnamed{diagrams//pbctb421new}
\InputIfFileExists{diagrams//pbctb421new.tikz}{}{\input{./figures/diagrams//pbctb421new.tikz}}
\endpgfgraphicnamed=\left(%
\beginpgfgraphicnamed{diagrams//pbctb422new}
\InputIfFileExists{diagrams//pbctb422new.tikz}{}{\input{./figures/diagrams//pbctb422new.tikz}}
\endpgfgraphicnamed\right)^{-1}
\]
Thus:
\begin{equation}\label{pbctb52}
\beginpgfgraphicnamed{diagrams//pbctb52new}
\InputIfFileExists{diagrams//pbctb52new.tikz}{}{\input{./figures/diagrams//pbctb52new.tikz}}
\endpgfgraphicnamed
\end{equation}
 Using again the same technique as earlier we obtain:
%Therefore, $B=$
%$$
%\tikzfig{diagrams//pbctb62}
%$$
$$
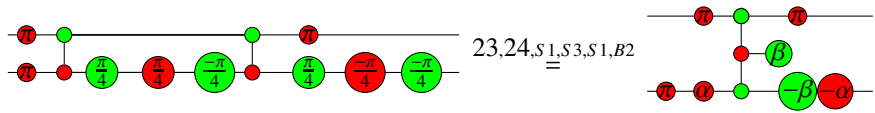
\beginpgfgraphicnamed{diagrams//pbctb62-reduced}
\InputIfFileExists{diagrams//pbctb62-reduced.tikz}{}{\input{./figures/diagrams//pbctb62-reduced.tikz}}
\endpgfgraphicnamed
$$
Finally,  again following the previous lemma, $B^2=$
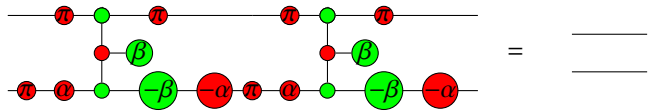
\[
\beginpgfgraphicnamed{diagrams//pbctb72new-reduced}
\InputIfFileExists{diagrams//pbctb72new-reduced.tikz}{}{\input{./figures/diagrams//pbctb72new-reduced.tikz}}
\endpgfgraphicnamed
\]
%$$
%\tikzfig{diagrams//pbctb72new}
%$$
\end{proof}

\begin{lemma}
Let $C=$
\[
\beginpgfgraphicnamed{diagrams//pbctc1}
\InputIfFileExists{diagrams//pbctc1.tikz}{}{\input{./figures/diagrams//pbctc1.tikz}}
\endpgfgraphicnamed
\]
and $D=$
\[
\beginpgfgraphicnamed{diagrams//pbctc2}
\InputIfFileExists{diagrams//pbctc2.tikz}{}{\input{./figures/diagrams//pbctc2.tikz}}
\endpgfgraphicnamed
\]
 then   $D\circ C=I$.
\end{lemma}

\begin{proof}
Firstly we simplify the circuit $C$ as follows:
\[
\beginpgfgraphicnamed{diagrams//pbctc3new}
\InputIfFileExists{diagrams//pbctc3new.tikz}{}{\input{./figures/diagrams//pbctc3new.tikz}}
\endpgfgraphicnamed
\]
By the rule (P), we can assume that: 
\begin{equation}\label{circuitceq}
\beginpgfgraphicnamed{diagrams//pbctc32new}
\InputIfFileExists{diagrams//pbctc32new.tikz}{}{\input{./figures/diagrams//pbctc32new.tikz}}
\endpgfgraphicnamed
\end{equation}
Then we have  for $C$: % $C=$
\begin{equation}\label{circuitceq4}
\beginpgfgraphicnamed{diagrams//pbctc34new}
\InputIfFileExists{diagrams//pbctc34new.tikz}{}{\input{./figures/diagrams//pbctc34new.tikz}}
\endpgfgraphicnamed
\end{equation}
Secondly, we simplify the circuit $D$ as follows:
\begin{equation*}
\beginpgfgraphicnamed{diagrams//pbctc4new}
\InputIfFileExists{diagrams//pbctc4new.tikz}{}{\input{./figures/diagrams//pbctc4new.tikz}}
\endpgfgraphicnamed
\end{equation*}
By the rule (P), we have
\begin{equation}\label{circuitceq2}
\beginpgfgraphicnamed{diagrams//pbctc33new}
\InputIfFileExists{diagrams//pbctc33new.tikz}{}{\input{./figures/diagrams//pbctc33new.tikz}}
\endpgfgraphicnamed
\end{equation}
Therefore  we have for $D$: 
\begin{equation}\label{circuitceq5}
\beginpgfgraphicnamed{diagrams//pbctc35new}
\InputIfFileExists{diagrams//pbctc35new.tikz}{}{\input{./figures/diagrams//pbctc35new.tikz}}
\endpgfgraphicnamed
\end{equation}
Then we obtain the composition  for $D\circ C=$
\begin{equation}\label{circuitceq6}
\beginpgfgraphicnamed{diagrams//pbctc36new}
\InputIfFileExists{diagrams//pbctc36new.tikz}{}{\input{./figures/diagrams//pbctc36new.tikz}}
\endpgfgraphicnamed
\end{equation}
%By Corollary \ref{zxztoxzxcr},
By the rule (P),
 we can assume that:
\begin{equation}\label{circuitceq7}
\beginpgfgraphicnamed{diagrams//pbctc37new}
\InputIfFileExists{diagrams//pbctc37new.tikz}{}{\input{./figures/diagrams//pbctc37new.tikz}}
\endpgfgraphicnamed
\end{equation}
Then for its inverse, we have 
\begin{equation}\label{circuitceq8}
\beginpgfgraphicnamed{diagrams//pbctc38new}
\InputIfFileExists{diagrams//pbctc38new.tikz}{}{\input{./figures/diagrams//pbctc38new.tikz}}
\endpgfgraphicnamed
\end{equation}
Also we can obtain that:
\begin{equation}\label{circuitceq9}
\beginpgfgraphicnamed{diagrams//pbctc39new}
\InputIfFileExists{diagrams//pbctc39new.tikz}{}{\input{./figures/diagrams//pbctc39new.tikz}}
\endpgfgraphicnamed
\end{equation}
As a consequence, we have the inverse for both sides of (\ref{circuitceq9}):
\begin{equation}\label{circuitceq10}
\beginpgfgraphicnamed{diagrams//pbctc310new}
\InputIfFileExists{diagrams//pbctc310new.tikz}{}{\input{./figures/diagrams//pbctc310new.tikz}}
\endpgfgraphicnamed
\end{equation}
Now we can rewrite  $D\circ C$ as: 
%\[
%\tikzfig{diagrams//pbctc311new-part1}
%\]
%\begin{equation}\label{circuitceq11}
%\tikzfig{diagrams//pbctc311new-part2}
%\end{equation}
\begin{equation}\label{circuitceq11}
\beginpgfgraphicnamed{diagrams//pbctc311new}
\InputIfFileExists{diagrams//pbctc311new.tikz}{}{\input{./figures/diagrams//pbctc311new.tikz}}
\endpgfgraphicnamed
\end{equation}
We can depict the dashed part of (\ref{circuitceq11}) in a form of connected octagons, and to deal with these octagons  we use (Hex):
%\[
%\tikzfig{diagrams//octagon}%\tikzfig{diagrams//octagonsimpl}
%\]
%\begin{equation}\label{octagon1eq}
%\tikzfig{diagrams//octagon1}
%\end{equation}
%\[
%\tikzfig{diagrams//octagon1}
%\]
%we need to cope with hexagons.
%\begin{equation}\label{hexagon1eq}
%\tikzfig{diagrams//hexagon1new}
%\end{equation}
%%\[
%%\tikzfig{diagrams//hexagon1}
%%\]
%By colour change rule (H2), 
%we have: \e
% Then we can rewrite   (\ref{octagon1eq}) as follows:
%\begin{equation}\label{octagon2checkeq}
%\tikzfig{diagrams//octagon2checknew}
%\end{equation}
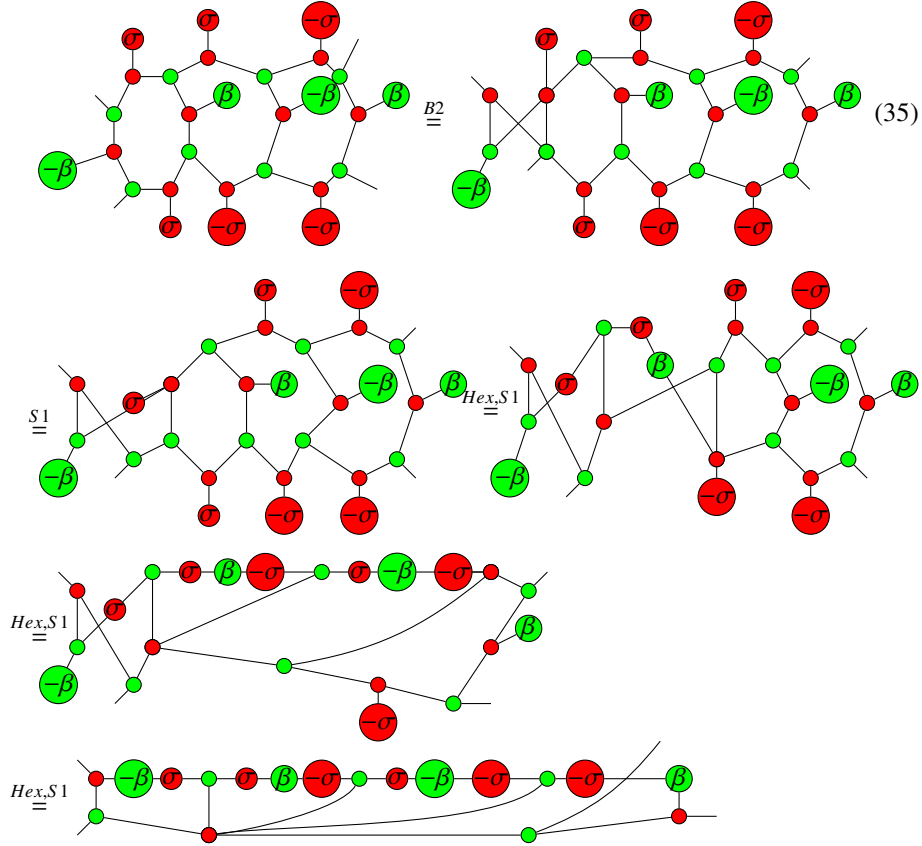
\begin{equation}\label{octagon2checkeq}
\beginpgfgraphicnamed{diagrams//octagon2checknew-part1}
\InputIfFileExists{diagrams//octagon2checknew-part1.tikz}{}{\input{./figures/diagrams//octagon2checknew-part1.tikz}}
\endpgfgraphicnamed
\end{equation}
\[%\begin{equation}\label{octagon2checkeq}
\beginpgfgraphicnamed{diagrams//octagon2checknew-part2}
\InputIfFileExists{diagrams//octagon2checknew-part2.tikz}{}{\input{./figures/diagrams//octagon2checknew-part2.tikz}}
\endpgfgraphicnamed
\]%\end{equation}
%\tikzfig{diagrams//octagon2}
%where we used property of hexagons in the third, fourth and fifth equalities.
%\\[tikzfig{diagrams//octagonsimplcorrect}\]
By the (P) rule, we have: 
\begin{equation}\label{octasim2eq}
\beginpgfgraphicnamed{diagrams//octasim2new}
\InputIfFileExists{diagrams//octasim2new.tikz}{}{\input{./figures/diagrams//octasim2new.tikz}}
\endpgfgraphicnamed
\end{equation}
where $z=x+\pi$.  Then we take inverse for each side of (\ref{octasim2eq}) and obtain that:
\begin{equation}\label{octasim3eq}
\beginpgfgraphicnamed{diagrams//octasim3new}
\InputIfFileExists{diagrams//octasim3new.tikz}{}{\input{./figures/diagrams//octasim3new.tikz}}
\endpgfgraphicnamed
\end{equation}
By rearranging the phases on both sides of (\ref{octasim2eq}), we have:
\begin{equation}\label{octasim4eq}
\beginpgfgraphicnamed{diagrams//octasim4new}
\InputIfFileExists{diagrams//octasim4new.tikz}{}{\input{./figures/diagrams//octasim4new.tikz}}
\endpgfgraphicnamed
\end{equation}
Thus:
\begin{equation}\label{octasim5eq}
\beginpgfgraphicnamed{diagrams//octasim5new}
\InputIfFileExists{diagrams//octasim5new.tikz}{}{\input{./figures/diagrams//octasim5new.tikz}}
\endpgfgraphicnamed
\end{equation}
Therefore:
\begin{equation}\label{octasim6eq}
\beginpgfgraphicnamed{diagrams//octasim6new}
\InputIfFileExists{diagrams//octasim6new.tikz}{}{\input{./figures/diagrams//octasim6new.tikz}}
\endpgfgraphicnamed
\end{equation}
%By  (\ref{octasim6eq}), we have 
It then follows that:
\begin{equation}\label{octasim7eq}
\beginpgfgraphicnamed{diagrams//octasim7new}
\InputIfFileExists{diagrams//octasim7new.tikz}{}{\input{./figures/diagrams//octasim7new.tikz}}
\endpgfgraphicnamed
\end{equation}
If we take the inverse of the left-hand-side of (\ref{octasim7eq}), then we have:
\begin{equation}\label{octasim8eq}
\beginpgfgraphicnamed{diagrams//octasim8new}
\InputIfFileExists{diagrams//octasim8new.tikz}{}{\input{./figures/diagrams//octasim8new.tikz}}
\endpgfgraphicnamed
\end{equation}
Now we can further simplify the final diagram in  (\ref{octagon2checkeq}) as follows:
\begin{equation}\label{octasim9eq}
\beginpgfgraphicnamed{diagrams//octasim9new}
\InputIfFileExists{diagrams//octasim9new.tikz}{}{\input{./figures/diagrams//octasim9new.tikz}}
\endpgfgraphicnamed
\end{equation}
Finally, the composite circuit $D\circ C$ as  can be simplified as follows: 
\begin{equation}\label{octasim10eq}
\beginpgfgraphicnamed{diagrams//octasim10new}
\InputIfFileExists{diagrams//octasim10new.tikz}{}{\input{./figures/diagrams//octasim10new.tikz}}
\endpgfgraphicnamed
\end{equation}
%where we used (\ref{circuitceq9}) to get 
where we used the following property:
\begin{equation}\label{octasim11eq}
\beginpgfgraphicnamed{diagrams//octasim11new}
\InputIfFileExists{diagrams//octasim11new.tikz}{}{\input{./figures/diagrams//octasim11new.tikz}}
\endpgfgraphicnamed
\end{equation}
\end{proof}

\section{Conclusion and further work}

We gave a set of ZX-rules that allows one to establish all equations between 2-qubit circuits, and these ZX-rules are remarkably simpler than the relations between unitary gates from which they were derived.  The key to this simplicity is: (i) abandoning unitarity at intermediate stages, and (ii) abandoning the T-restriction, which comes about when applying rule (P).   In the case of the latter, it is important to stress again that the actual values of the phases in the RHS of (P) don't have to be known.

Also, while the techniques used to establish the relations between two-qubit unitary gates don't scale to more than two qubits, the ZX-calculus, by being complete, already provides us with such a set.  It is just a matter to figure out if all of those rules are actually needed for the case of circuits.  Automation is moreover also possible thanks to the quantomatic software.  Although we don't  yet have a general strategy for simplifying quantum circuits by the ZX-calculus, it is possible at least in some cases.    In fact, in ongoing work in collaboration with Niel de Beaudrap, using similar techniques as some of the ones in this paper, we have shown that using ZX-calculus we can outperform the state-of-the-art for quantum circuit simplification.  A paper on this is forthcoming.

We expect the new rule (P) to have many more utilities within the domain of quantum computation and information.  The same question remains for other rules that emerged as part of the the completion of ZX-calculus.

A natural challenge of interest to the Reversible Computing community is whether the classical fragment of ZX-calculus can be used for deriving similar completeness results for classical circuits.

%\bR 
%... quantum circuit simplification ...
%
%... from qubits to qudits ...
%
%We are also able to give an analytic solution for converting from ZXZ to XZX Euler decompositions of single-qubit unitary gates as suggested by Schr\"oder de Witt and Zamdzhiev.  Our result in this paper is an important step towards efficient simplification of arbitrary n-qubit  Clifford+T circuits.
%
%There are several questions for the next step. Firstly, can we derive the completeness of ZX-calculus for 2-qubit Clifford+T circuits from the universal completeness? Secondly, can we obtain the completeness of ZX-calculus for 3-qubit Clifford+T circuits,  or arbitrary n-qubit Clifford+T circuits?  
%
%It is also interesting to incorporate the rules of ZX-calculus for 2-qubit Clifford+T circuits in the automated graph rewriting system Quantomatic  \cite{Quanto}. 
%
%\e

\section*{Acknowledgement}

This work was sponsored by Cambridge Quantum Computing Inc.~for which we are grateful. QW also thanks   Kang Feng Ng for useful discussions.  

\bibliographystyle{splncs03}
\bibliography{completeness2qubit}     

%\iffalse
\newpage\appendix

\section{ Verification of the complete relations in the ZX-calculus}
Firstly, we explain how the ZX rule (P) is obtained.  By \emph{generalised phases} we mean:
\[
\beginpgfgraphicnamed{diagrams//greenbxa}
\begin{tikzpicture}
	\begin{pgfonlayer}{nodelayer}
		\node [style=none] (0) at (0, 0.25) {};
		\node [style=gbox] (1) at (0, 0) {$a$};
		\node [style=none] (2) at (0, -0.25) {};
	\end{pgfonlayer}
	\begin{pgfonlayer}{edgelayer}
		\draw (0.center) to (2.center);
	\end{pgfonlayer}
\end{tikzpicture}}
\endpgfgraphicnamed:= 
 \begin{pmatrix}
        1 & 0 \\
        0 & a \end{pmatrix}
\qquad \qquad        
\beginpgfgraphicnamed{diagrams//redbxb}
\InputIfFileExists{diagrams//redbxb.tikz}{}{\input{./figures/diagrams//redbxb.tikz}}
\endpgfgraphicnamed:= \begin{pmatrix}
        1+a & 1-a \\
        1-a & 1+a \end{pmatrix}.
\]
where $a$ is  an arbitrary complex number.

\begin{lemma}\label{phaseclcg}(Generalised phases colour-swap  law)
We have:
\begin{equation}\label{colorspps}
\beginpgfgraphicnamed{diagrams//gpcswp}
\InputIfFileExists{diagrams//gpcswp.tikz}{}{\input{./figures/diagrams//gpcswp.tikz}}
\endpgfgraphicnamed
\end{equation}
where:
$$
a_2=-i(U+V)\sqrt{\frac{S}{T}},\qquad
b_2=\frac{\tau+i\sqrt{\frac{T}{S}}}{\tau-i\sqrt{\frac{T}{S}}}
,\qquad c_2=-i(U-V)\sqrt{\frac{S}{T}}.
 $$
 with:
\begin{equation}\label{taouvst}
\begin{array}{ll}
\tau=  (1-b_1)(a_1+c_1)+(1+b_1)(1+a_1c_1),&\\
U=(1+b_1)(a_1c_1-1), \\
V=( 1-b_1)(a_1-c_1),\\
S=( 1-b_1)(a_1+c_1)-(1+b_1)(1+a_1c_1), \\
T=\tau(U^2-V^2).
\end{array}
\end{equation}
% $$S=\tau-2U -4(1+b_1) , ~~T=\tau(U^2-V^2).$$
Especially, if $a_1=c_1,$ then  $a_2=c_2$;  if $a_1c_1=1$, then  $a_2=-c_2$. 
\end{lemma}
\begin{proof}
The matrix of the left-hand-side of (\ref{colorspps}) is 
\[
\begin{pmatrix}
        1 & 0 \\
        0 & c_1
 \end{pmatrix}
 \begin{pmatrix}
       1+b_1 &  1-b_1 \\
        1-b_1 &  1+b_1
 \end{pmatrix}
 \begin{pmatrix}
        1 & 0 \\
        0 & a_1
 \end{pmatrix}
 =\begin{pmatrix}
         1+b_1 &   a_1( 1-b_1)\\
        c_1( 1-b_1) &  a_1 c_1( 1+b_1)
 \end{pmatrix}
 \]
 The matrix of the right-hand-hand-side of (\ref{colorspps}) is 
$$
 \begin{pmatrix}
       1+c_2 &  1-c_2 \\
        1-c_2 &  1+c_2
 \end{pmatrix}
 \begin{pmatrix}
        1 & 0 \\
        0 & b_2
 \end{pmatrix}
 \begin{pmatrix}
       1+a_2 &  1-a_2 \\
        1-a_2 &  1+a_2
 \end{pmatrix}
 $$
\[
=\begin{pmatrix}
      (1+c_2) (1+a_2) +(1-c_2) b_2(1-a_2)&  
       (1+c_2) (1-a_2) +(1-c_2) b_2(1+a_2) \\
         (1-c_2) (1+a_2) +(1+c_2) b_2(1-a_2) & 
           (1-c_2) (1-a_2) +(1+c_2) b_2(1+a_2)
 \end{pmatrix}
 :=\begin{pmatrix}
      X &  Y \\
        Z &  W
 \end{pmatrix}
\]
To let the equality (\ref{colorspps})  hold, there must exist a non-zero complex number $k$ such that
\begin{equation}\label{mtrixk}
\begin{pmatrix}
      X &  Y \\
        Z &  W
 \end{pmatrix}=
 k\begin{pmatrix}
         1+b_1 &   a_1( 1-b_1)\\
        c_1( 1-b_1) &  a_1 c_1( 1+b_1)
 \end{pmatrix}
\end{equation}
Then 
$$X+Y=2(1+b_2+c_2-b_2c_2)=k[1+b_1 +  a_1( 1-b_1)]$$
\[
Z+W=[(1-c_2)2+(1+c_2) b_22]=2(1+b_2-c_2+b_2c_2)=k[ c_1( 1-b_1)+a_1c_1(1+b_1) ]
\]
Thus
\[
X+Y+Z+W=4(1+b_2)=k[(1+b_1)(1+a_1c_1)+( 1-b_1)(a_1+c_1)], 
\]
i.e.,
\[
b_2=\frac{k}{4}[(1+b_1)(1+a_1c_1)+( 1-b_1)(a_1+c_1)]-1=\frac{k}{4}\tau -1, 
\]
and
\[
Z+W-(X+Y)=4c_2(b_2-1)=k[(1+b_1)(a_1c_1-1)+( 1-b_1)(c_1-a_1)], 
\]
i.e.,
\[
c_2=\frac{\frac{k}{4}[(1+b_1)(a_1c_1-1)+( 1-b_1)(c_1-a_1)]}{\frac{k}{4}\tau -2}=\frac{k[(1+b_1)(a_1c_1-1)+( 1-b_1)(c_1-a_1)]}{k\tau -8}.
\]
Similarly,
$$X+Z=2(1+a_2+b_2-a_2b_2)=k[1+b_1 +  c_1( 1-b_1)]$$
\[
Y+W=2(1+b_2-a_2+b_2a_2)=k[ a_1( 1-b_1)+a_1c_1(1+b_1) ]
\]

\[
Y+W-(X+Z)=4a_2(b_2-1)=k[(1+b_1)(a_1c_1-1)+( 1-b_1)(a_1-c_1)], 
\]
i.e.,
\[
a_2=\frac{\frac{k}{4}[(1+b_1)(a_1c_1-1)+( 1-b_1)(a_1-c_1)]}{\frac{k}{4}\tau -2}=\frac{k[(1+b_1)(a_1c_1-1)+( 1-b_1)(a_1-c_1)]}{k\tau -8},
\]
Now we decide the value of $k$.
Let $U=(1+b_1)(a_1c_1-1), ~~V=( 1-b_1)(a_1-c_1).$ Then
\[
a_2+c_2=\frac{2kU}{k\tau -8}, ~~a_2c_2=\frac{k^2(U^2-V^2)}{(k\tau -8)^2}, 
\]
Furthermore,
\[
X=(1+c_2) (1+a_2) +(1-c_2) b_2(1-a_2)=k(1+b_1),
\]
i.e.,
\[
1+a_2+b_2 +c_2 +a_2c_2-a_2b_2-b_2c_2+a_2b_2c_2=k(1+b_1),
\]
by rearrangement, we have
\[
(a_2 +c_2)(1 -b_2)+(1+b_2)(1+a_2c_2)=k(1+b_1).
\]
Therefore,
\[
\frac{2kU}{k\tau -8}(2 -\frac{k}{4}\tau)+\frac{k}{4}\tau(1+\frac{k^2(U^2-V^2)}{(k\tau -8)^2})=k(1+b_1).
\]
Divide by $k$ on both sides, then multiply by $(k\tau -8)^2$  on both sides,  we obtain a quadratic equation of $k$:
\[
2U(k\tau -8)(2 -\frac{k}{4}\tau)+\frac{1}{4}\tau[(k\tau -8)^2+k^2(U^2-V^2)]=(k\tau -8)^2(1+b_1).
\]
By rearrangement, we have
\[
(k\tau -8)^2[\tau-2U -4(1+b_1)]+k^2\tau(U^2-V^2)=0.
\]
Let 
\[
\begin{array}{ll}
S=\tau-2U -4(1+b_1)&=(1+b_1)(1+a_1c_1)+( 1-b_1)(a_1+c_1)-2[(1+b_1)(a_1c_1-1)+2(1+b_1)]\\
&=
( 1-b_1)(a_1+c_1)-(1+b_1)(1+a_1c_1),\\
T=\tau(U^2-V^2).&
\end{array}
\]
Then the equation can be rewritten as
\[
(S\tau^2+T)k^2-16S\tau k+64S=0.
\]
Solve this equation, we have
\[
k=\frac{8S\tau \pm 8\sqrt{-ST}}{S\tau^2+T}.
\]
When we calculate the square root, we will consider its sign, so here we  can just write $k$ as 
\[
k=\frac{8S\tau + 8\sqrt{-ST}}{S\tau^2+T}.
\]
Now 
\[
\frac{8}{k}=\frac{8}{\frac{8S\tau + 8\sqrt{-ST}}{S\tau^2+T}}=\frac{S\tau^2+T}{S\tau + \sqrt{-ST}}=\frac{(\sqrt{S}\tau+i\sqrt{T})(\sqrt{S}\tau-i\sqrt{T})}{\sqrt{S}(\sqrt{S}\tau+i\sqrt{T})}=\tau-i\sqrt{\frac{T}{S}},
\]
i.e.,
$$
k=\frac{8}{\tau-i\sqrt{\frac{T}{S}}}.
$$
Then
$$
a_2=\frac{k(U+V)}{k\tau-8}=\frac{U+V}{\tau-\frac{8}{k}}=\frac{U+V}{i\sqrt{\frac{T}{S}}}
=-i(U+V)\sqrt{\frac{S}{T}}.$$
$$
c_2=-i(U-V)\sqrt{\frac{S}{T}}, ~~b_2=\frac{\tau+i\sqrt{\frac{T}{S}}}{\tau-i\sqrt{\frac{T}{S}}}.
$$
\end{proof}

\begin{corollary}\label{zxztoxzxcr}
For regular phases we have: 
\begin{equation}
\beginpgfgraphicnamed{diagrams//zxztoxzx}
\InputIfFileExists{diagrams//zxztoxzx.tikz}{}{\input{./figures/diagrams//zxztoxzx.tikz}}
\endpgfgraphicnamed
\end{equation}
with $\alpha_1, \beta_1, \gamma_1 \in (0, ~2\pi), \alpha_2=\arg z+\arg z_1, \gamma_2=\arg z-\arg z_1,  \beta_2=2\arg (|\frac{z}{z_1}|+i)$,  where
 $z=\cos\frac{\beta_1}{2}\cos\frac{\alpha_1+\gamma_1}{2}+i\sin\frac{\beta_1}{2}\cos\frac{\alpha_1-\gamma_1}{2}$, $z_1=\cos\frac{\beta_1}{2}\sin\frac{\alpha_1+\gamma_1}{2}-i\sin\frac{\beta_1}{2}\sin\frac{\alpha_1-\gamma_1}{2}$, $z_2=|\frac{z}{z_1}|+i $. If $\alpha_1=\gamma_1$, then $\alpha_2=\gamma_2$;  If $\alpha_1=-\gamma_1$, then $\alpha_2=\pi+\gamma_2(Mod ~2\pi)$.
\end{corollary}
\begin{proof}
In (\ref{colorspps}), let $\lambda_1=e^{i\alpha_1}, \lambda_2=e^{i\beta_1}, \lambda_3=e^{i\gamma_1}.$  Then for the values of $ U, V, S, \tau$ in (\ref{taouvst}) we have
$$
\begin{array}{ll}
U=4ie^{i\frac{\alpha_1+\beta_1+\gamma_1}{2}}\cos\frac{\beta_1}{2}\sin\frac{\alpha_1+\gamma_1}{2}, &
V=4e^{i\frac{\alpha_1+\beta_1+\gamma_1}{2}}\sin\frac{\beta_1}{2}\sin\frac{\alpha_1-\gamma_1}{2},\\
S=4e^{i\frac{\alpha_1+\beta_1+\gamma_1}{2}}z, & \tau=4e^{i\frac{\alpha_1+\beta_1+\gamma_1}{2}}\overline{z},
\end{array}
$$
where $z=\cos\frac{\beta_1}{2}\cos\frac{\alpha_1+\gamma_1}{2}+i\sin\frac{\beta_1}{2}\cos\frac{\alpha_1-\gamma_1}{2}$, $\overline{z}$ is the complex conjugate of $z$. 
Also, if we let $z_1=\cos\frac{\beta_1}{2}\sin\frac{\alpha_1+\gamma_1}{2}-i\sin\frac{\beta_1}{2}\sin\frac{\alpha_1-\gamma_1}{2}$,  then 
$$
U+V=4ie^{i\frac{\alpha_1+\beta_1+\gamma_1}{2}}z_1, ~~U-V=4ie^{i\frac{\alpha_1+\beta_1+\gamma_1}{2}}\overline{z}_1.
$$
Thus 
$$
\frac{U+V}{U-V}=\frac{z_1}{\overline{z}_1}=\frac{z_1^2}{|z_1|^2},  ~~\sqrt{\frac{U+V}{U-V}}=\frac{z_1}{|z_1|}=e^{i\theta}.
$$
where $|z_1|$ is the magnitude of the complex number $z_1$, and $\theta=\arg z_1\in [0, 2\pi)$ is the phase of $z_1$.
Similarly, we have 
$$
\sqrt{\frac{z}{\overline{z}}}=\frac{z}{|z|}=e^{i\phi},
$$ 
where $\phi=\arg z\in [0, 2\pi)$ is the phase of $z$.
Therefore,
$$
\begin{array}{ll}
\sigma_1=-i(U+V)\sqrt{\frac{S}{T}}=-i(U+V)\sqrt{\frac{S}{\tau(U+V)(U-V)}}&\\
=-i\sqrt{\frac{S}{\tau}\frac{(U+V)}{(U-V)}}=-i\sqrt{-\frac{z}{\overline{z}}\frac{z_1}{\overline{z_1}}}=-iie^{i\phi}e^{i\theta}=e^{i(\phi+\theta)},
\end{array}
$$
$$
\begin{array}{ll}
\sigma_3=-i(U-V)\sqrt{\frac{S}{T}}=-i(U-V)\sqrt{\frac{S}{\tau(U+V)(U-V)}}&\\
=-i\sqrt{\frac{S}{\tau}\frac{(U-V)}{(U+V)}}=-i\sqrt{-\frac{z}{\overline{z}}\frac{\overline{z_1}}{z_1}}=-iie^{i\phi}e^{-i\theta}=e^{i(\phi-\theta)},
\end{array}
$$
$$
\begin{array}{ll}
\sigma_2=\frac{\tau+i\sqrt{\frac{T}{S}}}{\tau-i\sqrt{\frac{T}{S}}}=\frac{\tau\sqrt{\frac{S}{T}}+i}{\tau\sqrt{\frac{S}{T}}-i}=\frac{\sqrt{\frac{S\tau^2}{\tau(U^2-V^2)}}+i}{\sqrt{\frac{S\tau^2}{\tau(U^2-V^2)}}-i}\\
=\frac{\sqrt{\frac{S\tau}{(U+V)(U-V)}}+i}{\sqrt{\frac{S\tau}{(U+V)(U-V)}}-i}=\frac{\sqrt{\frac{z\overline{z}}{z_1\overline{z_1}}}+i}{\sqrt{\frac{z\overline{z}}{z_1\overline{z_1}}}-i}=\frac{|\frac{z}{z_1}|+i}{|\frac{z}{z_1}|-i}=\frac{z_2}{\overline{z_2}}=(\frac{z_2}{|z_2|})^2=e^{i2\varphi},
\end{array}
$$
where $z_2=|\frac{z}{z_1}|+i, ~\varphi=\arg z_2$ is the phase of $z_2$.  Let $\alpha_2=\phi+\theta,  \beta_2=2\varphi, \gamma_2=\phi-\theta$.
Apparently, if $\alpha_1=\gamma_1$, then $V=0,$ i.e., $e^{i\theta}=1,$ thus $\theta=0$. It follows that 
$\alpha_2=\gamma_2$. If $\alpha_1=-\gamma_1$, then $U=0,$ thus
 $e^{i(\phi+\theta)}=-e^{i(\phi-\theta)}=e^{i(\pi+\phi-\theta)}$, i.e.,  $e^{i\alpha_2 }=e^{i(\pi+\gamma_2)}$. 
 %Thus $\alpha_2\equiv \pi+\gamma_2(Mod ~2\pi)$.
 Thus $\alpha_2= \pi+\gamma_2$. 
\end{proof}

%\fi

%\begin{thebibliography}{99}
%
%\bibitem{CoeckeDuncan} B. Coecke,  R. Duncan (2011): Interacting quantum
%observables: Categorical algebra and diagrammatics. New Journal of Physics 13, p.
%043016.
%
%\bibitem{ptbian} Peter Selinger and Xiaoning Bian, Relations for Clifford+T operators on two qubits, Quantum Programming and Circuits Workshop, June 2015. 
%
%\bibitem{Zamdzhiev} C. Schr\"oder de Witt, V. Zamdzhiev. The ZX-calculus is incomplete for quantum mechanics. EPTCS 172, pp.285-292, 2014.
%
%\bibitem{Miriam1ct}Miriam Backens, The ZX-calculus is complete for the single-qubit Clifford+T group, Electronic Proceedings in Theoretical Computer Science, 2014.
%
%\bibitem{Emmanuel} Emmanuel Jeandel, Simon Perdrix, Renaud Vilmart, A Complete Axiomatisation of the ZX-Calculus for Clifford+T Quantum Mechanics, arXiv:1705.11151
%
%\bibitem{bpw} Miriam Backens, Simon Perdrix, and Quanlong Wang, A simplified stabilizer ZX-calculus, Electronic Proceedings in Theoretical Computer Science, 2016.
%
%\bibitem{ngwang} Kang Feng Ng and Quanlong Wang, A universal completion of the ZX-calculus. 2017, arXiv: 1706.09877.
%
%\bibitem{Quanto} Quantomatic. https://sites.google.com/site/quantomatic/
%
%\end{thebibliography}

\end{document}